\documentclass[11pt]{article}

\usepackage[text={6.5in,9in},centering]{geometry}
\usepackage[cmex10]{amsmath}
\usepackage{enumerate}
\usepackage{graphicx}
\usepackage{amsfonts}
\usepackage{amsthm}
\usepackage{url}
\usepackage{gensymb}
\usepackage{pdfpages}

\newtheorem{proposition}{Proposition}
\newtheorem{lemma}{Lemma}
\newtheorem{corollary}{Corollary}
\newtheorem{theorem}{Theorem}

\newtheorem{claim}{Claim}

\newcommand{\tlog}{\log}

\newcommand{\cG}{\mathcal G}

\newcommand{\cI}{\mathcal I}
\newcommand{\cS}{\mathcal S}

\begin{document}

\begin{titlepage}

\title{Wireless Aggregation at Nearly Constant Rate}

\author{
  Magn\'us M. Halld\'orsson
  \qquad
  Tigran Tonoyan \\ \\
  ICE-TCS, School of Computer Science \\
  Reykjavik University \\
  \url{mmh@ru.is, ttonoyan@gmail.com}
}

\maketitle
\thispagestyle{empty}

\begin{abstract}
  One of the most fundamental tasks in sensor networks is the computation of a (compressible) aggregation function of
  the input measurements. What rate of computation can be maintained, by properly choosing the aggregation tree, the
  TDMA schedule of the tree edges, and the transmission powers?  This can be viewed as the convergecast capacity of a
  wireless network.

  We show here that the optimal rate is effectively a constant.  This holds even in \emph{arbitrary} networks, under the
  physical model of interference.  This compares with previous bounds that are logarithmic (e.g., $\Omega(1/\log n)$).
  Namely, we show that a rate of $\Omega(1/\log^* \Delta)$ is possible, where $\Delta$ is the length diversity (ratio
  between the furthest to the shortest distance between nodes). 
  It also implies that the \emph{scheduling complexity} of wireless connectivity is $O(\log^* \Delta)$.
  This is achieved using the natural minimum spanning tree (MST). 
  Our method crucially depends on choosing the appropriate power assignment for the instance at hand,
  since without power control, only a trivial linear rate can be guaranteed.
  We also show that there is a fixed power assignment that allows for a rate of $\Omega(1/\log\log \Delta)$.

  Surprisingly, these bounds are best possible. No aggregation network can guarantee a rate better than $O(1/\log\log \Delta)$ using fixed power assignment. Also, when using arbitrary power control, there are instances whose 
MSTs  cannot be scheduled in fewer than $\Omega(1/\log^* \Delta)$ slots.
\end{abstract}

\end{titlepage}

\begin{figure*}[b!]
\centering
\includegraphics[width=\textwidth]{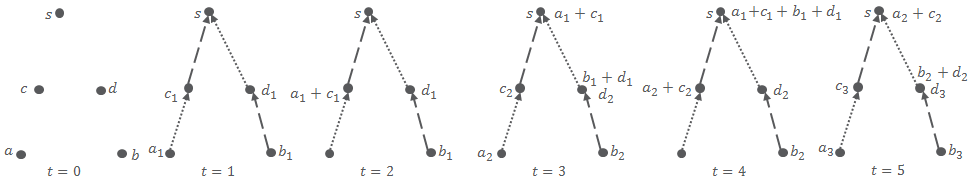}
\caption{Aggregation network consisting of five nodes, a spanning tree and a periodic schedule $S_1,S_2,S_1,S_2,\dots$, where $S_1$ is the pair of finely dashed links and $S_2$ is the pair of coarsely dashed links.}
\label{fig:aggregexample}
\end{figure*}

\section{Introduction}

Data collection is the primary task of most wireless sensor networks,
with the collected data commonly aggregated and compressed on its way to the sink. 
This may involve measurements taken by all the nodes at continuous rate, where we are interested in an aggregate property of the measurement, such as ``what is the maximum value?''
We are concerned with the fundamental question of the aggregation capacity of sensor networks: how fast can the information harvested from sensors be aggregated at the sink?

The aggregation problem has three components: choosing the \emph{aggregation tree} (or other connected graph spanning the set of nodes), selecting \emph{power assignment} of the sensor nodes and \emph{scheduling}  the communication links in the tree. By selecting a tree and a power assignment, we essentially define the space of  \emph{feasible} subsets of links in the tree, namely the subsets of the links that can be scheduled in a single time slot, without interfering with each other. Then, an \emph{aggregation schedule} is simply an infinite (or long enough, or periodic) sequence of feasible, each specifying the links that transmit in the corresponding time slot. 

Consider the network portrayed in Fig.~\ref{fig:aggregexample}. 
The sensor nodes $a,b,c,d$, shown on the left, are arranged into a tree as shown.
The links of the tree are shown with arrows, and they are assumed to interfere only when they share an endpoint.
We omit here the power assignment.
The schedule used is a periodic sequence $S_1,S_2,S_1,S_2,\dots$,
where the feasible set $S_1$ consists of the pair of finer dashed links and $S_2$ the coarser ones.
Measurements are generated at each node in odd-numbered time slots, e.g.\ $a_1, a_2, a_3$ at node $a$ in time slots $1,3,5$,
which are then forwarded and aggregated up the tree.
Each (measurement) \emph{frame} $i$ consists of the readings $(a_i,b_i,c_i,d_i)$.
We seek a sum aggregation of the frames at the sink, i.e., the sink should learn the values $a_i+b_i+c_i+d_i$, for each $i=1,2,\ldots$.

Following Fig.~\ref{fig:aggregexample} in more detail, we see that 
at the second time slot, reading $a_1$ has arrived at node $c$ where it is combined with the reading $c_1$ as the aggregated value $a_1 + c_1$. At time 3, $a_1+c_1$ has been forwarded to the sink while $b_1$ has been forwarded to $d$.
Furthermore, the second frame has arrived at  the sensor nodes. That means that node $d$ has two values in its buffer: $b_1 + d_1$ as well as the new reading $d_2$. Ultimately, the first frame will be aggregated at the root by start of timeslot 4, for a \emph{latency} of 3. Since the number of frames handled is half the number of time slots, this schedule attains a throughput \emph{rate} of $1/2$. It should be clear that a higher rate cannot be sustained, as it would lead to buffers overflowing.

In general, an aggregation schedule achieves a \emph{rate} (or \emph{scheduling complexity}) $\rho$ if each link is scheduled, on average, in every $1/\rho$-th time slot.
The \emph{aggregation capacity} of a given network (with a fixed aggregation tree and fixed or controllable power assignment) 
is the maximum aggregation rate achievable by any schedule, and the aggregation capacity of a pointset is the 
the maximum capacity achievable by selecting an appropriate tree, power assignment and schedule.

Much in the spirit of Gupta and Kumar~\cite{guptakumar}, the aggregation problem is usually addressed by deriving capacity scaling laws, where the aggregation rate is expressed in terms of the size of the network. Typically, scaling laws have been obtained for \emph{uniformly random} network deployments. Prior to the present work,  only logarithmic bounds have been  known for aggregation capacity in such networks (without using special techniques or properties, such as block coding). For instance, it is known that for uniformly distributed networks, the optimal rate is essentially $\Theta(1/\log n)$ when considering the protocol model of interference or when no power control is used, where $n$ is the number of nodes in the wireless network (see the Related Work for details).

A more recent thread of theoretical research, however, derives scaling laws for \emph{arbitrary network topologies}. It turns out that using appropriate power control and choosing an appropriate aggregation tree, one can obtain worst-case upper bounds that are comparable to the bounds obtained for uniformly distributed networks.

\paragraph*{Our Contribution} 
We show that near-constant aggregation rate is achievable on \emph{every} network instance in the physical model. 
More specifically, that rate is $\Omega\left(1/\log^*\Delta\right)$,
where $\log^*$ is extremely slow-growing function.
Here, $\Delta$ is the length diversity, or the ratio between the furthest to the shortest distance between nodes,
which is in all reasonable situation at most polynomial in $n$, the number of nodes.

The simplicity of the tree construction is an important feature: it is the \emph{minimum spanning tree} (MST) of the point set (directed arbitrarily). 
This is beneficial since the MST uses after all the shortest links available, implying energy efficiency and robustness.
Another important aspect is that our schedules are compactly represented as colorings, or partitions of the links.

In order to obtain the bound above, (near) optimal power control should be used, in which the power level of a node may depend on the power settings of other nodes. In other words, global power control algorithm should be used. As an alternative, we also consider special, \emph{oblivious} power assignments that depend only on the \emph{local} information -- the length of the communication link 
and show that aggregation rate of $\Omega\left(1/\log\log\Delta\right)$ can be achieved with such oblivious power.

Our worst-case bounds even improve on known average-case bounds. 
In networks where the nodes are uniformly distributed in a region of the plane (and, in fact, under most stochastic distributions), the parameter $\Delta$ is polynomial in $n$, with high probability. Thus, while best previous aggregation rate in this setting was $\Omega(1/\log n)$ \cite{GiridharK05}, 
we obtain bounds $\Omega\left(1/\log^* n\right)$ with global power control and $\Omega\left(1/\log\log n\right)$ with oblivious power schemes.
To our knowledge, this is the first case where analysis of worst-case deployments directly leads to improved results for random networks. 

Perhaps surprisingly, we show that our bounds are best possible. 
Specifically, $\Omega(1/\log\log \Delta)$ is the best (worst-case) rate possible when using any oblivious power assignment.
The bounds for arbitrary power are also best possible for the minimum spanning tree of a pointset:
We show that our analysis is tight in that the MST cannot lead to better aggregation rates:
we construct instances for which any power assignment and any schedule of the MST yields a rate of $O(1/\log^* \Delta)$.

Our results build on the approximation framework of \cite{us:stoc15,us:fsttcs15} that captures the additive SINR interference with appropriate unweighted graphs.
This results in very simple scheduling algorithms: they essentially consist of coloring a conflict graph. This also has algorithmic advantage due to the local decision whether a color is valid or not, which contrasts with the appearance of the physical model as being inherently non-local.
Our results are obtained by combining these and other non-trivial techniques developed in recent years to formulate and analyze algorithms in the physical model \cite{KesselheimSODA11,HMSODA12,GoussevskaiaHW14,HB15}. While this allows for relatively compact presentation,
it should not be confused with easy application of standard techniques.

\paragraph*{Related Work}
Data collection/aggregation, being an important part of wireless sensor networks, have been extensively studied in a variety of different settings. Here we review the closely related literature.

Research on scaling laws of network capacity originates in the work of Gupta and Kumar~\cite{guptakumar}, where the transport capacity of random networks was considered. The early work on scaling laws for aggregation capacity in random networks includes, e.g.~\cite{MarcoDLN03, GiridharK05}. It was shown, essentially, that in the \emph{protocol model}, the aggregation rate of random networks is $\Omega(1/\log n)$ (without using  coding techniques). Similar results for random networks are obtained also using percolation theory, e.g.~\cite{DousseFMMT06,Vaze12}, where it was shown that with a uniform power assignment, there is a coloring of the nodes with $O(\log n)$ colors (which induces a coloring of edges) that gives a connected network, and that $\Omega(\log n)$ colors are necessary. It is also known that in regular grids, constant aggregation rate can be achieved~\cite{AvinLPP12}. 

The problem for worst-case networks was first considered in~\cite{MoscibrodaW06}. It was shown that unlike the random networks, the worst-case aggregation capacity depends crucially on the power control algorithm, and that the aggregation rate  of some networks is $O(1/n)$ in the protocol model, or in the physical model  with no power control. In contrast, with appropriate power control, the aggregation rate can be improved exponentially, namely, to $\Omega(1/\log^4 n)$. This was soon improved to $\Omega(1/\log^2 n)$ \cite{MoscibrodaWZ06, MoscibrodaIPSN07}, and later to $\Omega(1/\log n)$~\cite{HMSODA12}. The latter is a tight bound on the \emph{latency} of the aggregation, while we focus here on the sustained \emph{throughput}.
Aggregation capacity problem for worst-case networks has also been studied in~\cite{IncelGKC12}, but their heuristics are mainly evaluated through experiments.

The work on aggregation capacity has also been extended to various settings where particular features of networks or aggregation functions come into  play. Examples are the works considering the generalized SINR model~\cite{XuLS13,WangJLLT14}, coding techniques (i.e., for special aggregation functions)~\cite{SubramanianGS07,GiridharK05}, MIMO \cite{FuQWL12,ZhengB07}, mobile sensor networks etc.
We refer the reader to~\cite{IncelGK11} for further bibliography on  aggregation/collection problems.

\paragraph*{Roadmap.} We present the main models and assumptions on the network, as well as the formal definitions and problem statement in Sec.~\ref{S:defs}. Sec.~\ref{S:aggregationprotocol} introduces the main framework and our aggregation protocols. In Sec.~\ref{S:impossibility}, we demonstrate tightness of our analysis of aggregation protocols, proving upper bounds on the worst-case aggregation capacity of some network instances.  Some technicalities are deferred to the Appendix.

\section{Model and Problem}
\label{S:defs}

\paragraph*{Network and Links} 

We model the wireless sensor nodes as a set $R$ of points, arbitrarily located on the Euclidean plane. The nodes transmit in the same frequency band, working in synchronized time slots, where a time slot is sufficient for communicating a single packet.

A \emph{communication link} (or simply link) represents a communication request from a sender node $s$ to a receiver node $r$. Whenever dealing with a fixed set $L$ of links, we assume they are numbered from $1$ to $n=|L|$ and each link $i$ has sender $s_i$ and receiver $r_i$.

We denote $d_{ij}=d(s_i,r_j)$\label{G:asymdistance} and $l_i=d(s_i,r_i)$\label{G:li}, where $d(\cdot, \cdot)$ denotes the Euclidean distance. $l_i$ is called the \emph{length} of link $i$. 
We let $\Delta(L)$\label{G:delta} denote the ratio between the longest and the shortest link lengths in $L$, and drop $L$ when clear from context.
For sets $S_1, S_2$ of links, we let $d(S_1,S_2)$ denote the minimum distance between \emph{a node} in $S_1$ and a
node in $S_2$. In particular, $d(i,j)$ denotes the minimum distance between the nodes of two links $i,j$.

Given a set $S$ of links, we let $S_i^+=\{j\in S : l_j\ge l_i\}$\label{G:liplus} denote the subset of links that are longer than link $i$,
and similarly $S_i^-=\{j\in S : l_j\le l_i\}$\label{G:liminus} the subset of links shorter than $i$.

\paragraph*{Power Assignments} 
A \emph{power assignment} for a set $L$ of links is a function $P:L\rightarrow \mathbb{R}_+$. For each link $i$, $P(i)$\label{G:power} defines the power level used by the sender node $s_i$. 

We consider two modes of power control: \emph{global power control} and \emph{oblivious power schemes}. 
In the former, the nodes have the possibility to choose their power level arbitrarily, possibly taking into account the power levels of all other nodes. This assumption may be too strong in certain scenarios. In such cases, oblivious power schemes  $P_{\tau}$\label{G:powertau} of the form $P_{\tau}(i)=C\cdot l_i^{\tau\alpha}$ may be used, where $C$ is constant for the given network instance and $\tau\ge 0$. Note that in such power assignments, the power level of each link depends only on a local information - the link length. The simplest of such power schemes are the  \emph{uniform} power scheme ($P_0$) and  \emph{linear} power scheme ($P_1$).

\paragraph*{Feasibility}
Since the nodes transmit in the same frequency band, there is potential interference between parallel transmissions, so not every set of links can transmit at the same time. In order to capture the sets that a single time slot can accommodate, we use the notion of \emph{feasibility} of sets of links, modeled by the \emph{physical model} of communication. In the \emph{physical model} of communication, when using a power assignment $P$, a transmission of a link $i$ is successful if and only if  
\begin{equation}\label{E:sinr}
\cS_i\ge \beta\cdot \left(\sum_{j\in S\setminus \{i\}}{\cI_{ji}} + N\right),
\end{equation}
where $\cS_i$ denotes the received signal of link $i$, $\cI_{ji}$ denotes the interference on link $i$ caused by link $j$,   
$N\ge 0$\label{G:noise} is a constant denoting the ambient noise, $\beta>0$\label{G:beta} is the minimum SINR (Signal to Interference and Noise Ratio) required for a message to be successfully received and $S$ is the set of links transmitting concurrently with link $i$. We model signal attenuation through log-distance path-loss, which implies that  $\cS_i=\frac{P(i)}{l_{i}^\alpha}$ and $\cI_{ji}=\frac{P(j)}{d_{ji}^\alpha}$ for a constant $\alpha>2$.

A set $L$ of links is called $P$-\emph{feasible} if the condition~(\ref{E:sinr}) holds for each link $i\in L$ when using power assignment $P$. We say $L$ is \emph{feasible} if there exists a power assignment $P$ for which $L$ is $P$-feasible. Similarly, a collection of sets is $P$-feasible/feasible if each set in the collection is.

\paragraph*{Interference-limited networks}
We assume that set of links considered (specifically, the minimum spanning tree of the input pointset) is \emph{interference limited}. Namely, for each link $i$,  $P(i) \ge (1+\epsilon)\cdot \beta N l_i^\alpha$, where $\epsilon >0$ is a constant. Note that $P(i) = \beta N l_i^\alpha$ is the minimum power required for communicating over link $i$ in the absence of other transmissions. The assumption above implies that  the noise term can be ignored, i.e.\ setting $N=0$ affects only the constant factors in our results (see e.g. \cite{us:stoc15, us:fsttcs15}).

\paragraph*{Aggregation schedules and rate maximization}
An \emph{aggregation schedule} is a sequence $I_1, I_2, \ldots$ of feasible sets of links
such that the links $\cup_i I_i$ induce an acyclic digraph directed towards a given sink and spanning the pointset $R$.
The \emph{(aggregation) rate} of the schedule is the largest value $\rho$ such that
for all sufficiently large windows, $\exists m_0, \forall m > m_0$, 
each link appears at least $\rho m$ times in the first $m$ feasible sets.
The \emph{rate maximization} problem is to find an aggregation schedule of maximum rate for the given pointset.

Our positive results for rate maximization are obtained with a specific type of aggregation schedules, which are periodic repetitions of a \emph{coloring} of the link set.
Namely, a partition of a linkset $L$ into feasible ($P$-feasible) subsets is called a \emph{coloring (schedule)}. 
The \emph{number of colors/slots} or the \emph{schedule length} will refer to the number of subsets in the schedule,
and the rate then corresponds to the reciprocal of the schedule length.

\section{Aggregation Protocol}\label{S:aggregationprotocol}

Computing aggregation capacity involves three non-trivial cross-layer subtasks:
\begin{itemize}
\item \emph{Tree:} Selecting the edges of a converge-cast tree,
\item \emph{Power:} Choosing the power used by the transmitters,
\item \emph{Schedule:} Scheduling the transmissions using TDMA.
\end{itemize}
We are fortunately able to finesse the first two tasks, allowing us to focus entirely on the third.
Namely, the tree is simply the minimum spanning tree and the power can be a fixed function of the link length.
Even for the case of arbitrary power control where context-sensitive power is needed, we can leverage a formulation of Kesselheim \cite{KesselheimSODA11} that effectively take the power assignment out of the picture.

The scheduling task involves selecting feasible sets to transmit in each time slot, utilizing global synchronization.
In general, this can involve an infinite, non-repeating sequence of feasible sets.
Fortunately, a simpler approach succeeds here: we can find a short coloring schedule of the links of the tree (i.e. a partition into feasible subsets),
and periodically repeat this coloring.
The rate achieved is then inversely proportional to the length of the schedule.
Clearly, the rate achieved with a coloring schedule is a lower bound on the rate of an optimal schedule.

Our approach is to form a graph on the links, run a vertex coloring algorithm on this graph, and use the resulting coloring as a schedule of the links. Namely, the graph $G(T)$ contains a vertex for each link of the tree $T$, and we need to specify when there should be an edge between vertices corresponding to two tree-links.
We want this graph formulation to satisfy three properties:
\begin{enumerate}
 \item \emph{Feasibility}: Every independent set of $G(T)$ corresponds to a feasible set (a subset of $T$). Thus, a coloring of the nodes of $G(T)$ gives a valid coloring schedule of $T$. The length of the schedule is the number of colors used.
 \item \emph{Algorithmic tractability}: There is an efficient algorithm for coloring $G(T)$ with a constant-factor performance guarantee.
 \item \emph{Effectiveness}: The chromatic number of $G(T)$ should be small.
\end{enumerate}
All three properties imply that the length of the schedule will be short.

Graphs satisfying most of these properties were given in \cite{us:stoc15,us:fsttcs15}.
Namely, those graph formulations are such that for a given  set $L$ of links (not necessarily a tree),
there are graphs $G_{arb}(L), G_{obl}(L)$, and $G_1(L)$ such that
\begin{itemize}
 \item Independent sets of $G_{arb}(L)$ are feasible under arbitrary power control,
\item Independent sets of $G_{obl}(L)$ are feasible under an oblivious power assignment $P_\tau(i) \sim l_i^{\tau \alpha}$, where 
$l_i$ is the length of the link and $\tau \in (0,1)$ is a constant, and
\item The coloring problem on these graphs is constant-approximable, and
\item The schedules are short: \\
$\chi(G_{arb}(L)) = O(\log^* \Delta(L)) \cdot \chi(G_1(L))$ and \\
$\chi(G_{obl}(L)) = O(\log\log \Delta(L)) \cdot \chi(G_1(L))$.
\end{itemize}
For the reader's convenience, we give formal definitions of those graphs in the appendix.
All that is missing is the following piece of the puzzle that we supply in this work:
 \begin{quote}
 \emph{When $T$ is the MST of a planar pointset, $\chi(G_1(T)) = O(1)$.}
 \end{quote}
From this, our main results follow.
\begin{theorem}\label{T:schedulemst}
Let $R$ be any set of nodes in the plane and let $S$ be a set of links obtained by arbitrarily orienting the edges of an MST over $R$. Then the set $S$ has a coloring schedule of length  $O(\log^*{\Delta(S)})$ using a global power control algorithm, and a coloring schedule of length $O(\log\log{\Delta(S)})$ using an oblivious power scheme.
\end{theorem}

The scheduling algorithm is a simple and classic greedy algorithm for coloring the graph ($G_{arb}$ or $G_{obl})$):
Process the nodes/links in non-decreasing order of link-length and assign a node the smallest color that was not used to color its neighbors that preceded it.

The theorem implies
improved aggregation capacity bounds for networks distributed uniformly at random in a region of the plane, say a square of side $a>0$. Let $R$ be such a set. It follows readily from the results of e.g.~\cite[Sec. 4.4]{Moltchanov12}, that with high probability (say, $1-O(n^{-2})$), the minimum distance between any two nodes in $S$  is $a/poly(n)$. This implies that $\Delta = poly(n)$ w.h.p., implying the following result (note that we do not exclude that $a$ may depend on $n$).

\begin{corollary}
For a set of $n$ nodes distributed in a square of side $a>0$ or a disk of radius $a$ uniformly at random, the edges of the MST can be schedules in $O(\log^* n)$ slots using a global power control algorithm and in $O(\log\log{n})$ slots using an oblivious power scheme, with high probability.
\end{corollary}

More generally, the same property should hold for any non-heavy tailed distribution.

We give an overview of the graphs $G_{arb}(L)$ and $G_{obl}(L)$, and the arguments of \cite{us:stoc15,us:fsttcs15} in Appendix A,
but focus here on the graph $G_1$. But before going into detailed arguments and more technical results, we discuss several modeling and other  issues concerning the results above.

\subsection{Relevant Issues}


\paragraph*{Rate vs.\ latency} High rate and low latency are two desirable objectives that do not always go together.
Chains of unit-length links (or the regular grid) can be scheduled in constant number of slots, implying a constant rate,
but that causes the latency to be linear. Alternatively, by forming an appropriate tree, a latency of $O(\log n)$ can be attained \cite{HMSODA12}, but with a rate of $\Theta(1/\log n)$. The latency is, however, never better than $\log n$ (since in each round, at most half the nodes can forward their measurement to another node).
The focus in this paper is only on the optimal rate, but it is plausible that it can also lead to a bicriteria optimization with modest tradeoffs.

\paragraph*{Power limitations}
Our bounds apply also in the case when the nodes are power constrained. In this case, not all pairs of nodes can communicate (even without concurrent transmissions), but only the ones that are sufficiently close. This corresponds to a \emph{reduced} graph over the set of nodes, instead of the complete graph (not to be confused with conflict graphs). Now, it is sufficient to require that the maximum available power of the nodes is sufficient to communicate over the longest link of an MST of the reduced graph, i.e.\ that $P(i) \ge (1+\epsilon)\beta N l_i^\alpha$ still holds for all links of the MST. 
We refer to \cite{KesselheimESA12} and \cite{us:stoc15} for more details.
The latter assumption, which corresponds to the assumption of interference-limited networks, is a necessary one there are noise-limited networks for which only the trivial $1/n$-rate is possible.
Namely, if the nodes are barely reachable due to the noise, no spatial reuse may be possible.

\paragraph*{Multi-hop settings.} We assumed that all nodes are mutually reachable when using sufficient power.
This single-hop setting captures the intrinsic difficulty of organizing and scheduling communication in the physical model to overcome interference. It can fairly easily be extended to multi-hop settings by standard techniques:  selecting local leaders and performing flooding on the graph connecting those leaders. All the links used will be of roughly equal length, leading to analysis similar to that in the protocol model. The throughput of the flooding protocol is actually constant, and thus does not affect the performance of the combined procedure materially; see, e.g., \cite{BHM13P} for details.

\paragraph*{Other aggregation functions}
While we assumed that the aggregation function was fully compressible, our results can also aid the computation of other functions. For instance, in order to determine the \emph{median}, the typical approach is to apply binary search to count how many values fall above or under a given threshold. By applying our algorithm to each such counting aggregation, the combined complexity is correspondingly reduced.

\paragraph*{Pathloss assumptions} The assumptions of planarity of the pointset can be relaxed to more general doubling metric spaces, which may, e.g., be caused by shadowing effects. As shown in \cite{us:stoc15}, some such metric assumptions are necessary to improve on the at-least-logarithmic rate of uniform power. 

\paragraph*{Robustness and temporal variability}
Sporadic random fluctuations in noise or signal conditions do not have a significant effect, assuming an acknowledgment mechanism is included. The impact of Rayleigh fading, when the fading is independent across time (but not necessarily space), has also been shown to be minor \cite{dams2015}. A more extensive or long-term changes may naturally require repairing or reconstructing the tree and the schedule. 

\paragraph*{Approximations in terms of $n$}
In general, our results do not imply new bounds in terms of the number of nodes $n$ as theoretically, $\Delta$ and $n$ are independent parameters. It is, however, natural and common to assume that $\Delta$ is of at most polynomial growth in $n$, i.e., $\Delta \le n^{O(1)}$. Our results then imply rates of $\Omega(1/\log\log n)$ and $\Omega(1/\log^* n)$, even under the still weaker assumption of quasipolynomial growth, $\Delta \le n^{(\log n)^{O(1)}}$.

\subsection{Proof of the Key Theorem}
\label{sec:bounding}

We prove in this section the main claim about the chromatic number of the graph $G_1$.

The graph $G_1(L)$ is defined as follows.
Recall that $l_i$ denotes the length of link $i$ in $L$,
and $d(i,j)$ denotes the distance between the  closest points on the two links $i,j \in L$.
The links in $L$ are the vertices of $G_1$, and two links $i, j$ are adjacent in $G_1$ if $d(i,j) \le \min(l_i, l_j)$.

We will crucially use the following lemma from \cite{HMSODA12} which states a ``sparsity'' property of the MST. For links $i,j$, let us define the additive operator $I$ by:
\[
I(j,i)=\min\left\{1, \frac{l_j^\alpha}{d(i,j)^\alpha}\right\}\ .
\]
Let $I(S,j) = \sum_{i\in S} I(i,j)$ and $I(i,S) = \sum_{j \in S} I(i,j)$, for a set $S$ of links,
and recall that $S_i^+$ denotes the subset of links in $S$ that are longer than link $i$.
\begin{lemma}\cite[Lemma 4.2]{HMSODA12}\label{L:mainlemma}
Let $S$ be a set of links as in Theorem~\ref{T:schedulemst}. Then for any link $i\in S$, $I(i,S_i^+)=O(1)$.
\end{lemma}

\begin{theorem}
Let $T$ be the links of a MST induced by points in the plane.
Then, $\chi(G_1(T)) = O(1)$.
\label{thm:g1}
\end{theorem}

\begin{proof}
 As argued in~\cite{HMSODA12} (Lemma 3), it is possible to refine the set $T$
into a constant number of subsets $S^1,S^2,\dots,S^t$ such that for each set $S=S^k$ in the sequence and each link $i\in S$, it holds that $I(i,S_i^+) < 1$. This can be achieved by a simple first-fit algorithm, as follows. Initially, let  $S^k=\emptyset$ for all $k$. Iterate over the links in $T$ in a non-increasing order by length, assigning each link $i$ to the first set $S=S^k$ such that $I(i,S) < 1$. Note that at this point, $S$ consists of links that are not shorter than link $i$, i.e.\ $S=S_i^+$, and all links that are added to $S$ in subsequent iterations are no longer than $i$, which means that the sets $S^k$ have the desired property. It remains to show that $t\in O(1)$. This follows easily from the refinement procedure. Take any link $i\in S^t$, i.e.\ in the last set. Since it has been ``rejected'' by all previous sets $S^k$, we have that $I(i,T_i^+) > (t-1)\cdot 1$. But we have from Lemma~\ref{L:mainlemma} that $I(i,T_i^+) =O(1)$; hence, $t=O(1)$.

Let $S=S^k$ be one of the obtained subsets. It  suffices to prove the theorem only for $S$, as the number of such subsets is bounded by a constant. Thus, we have, by the construction, that for each link $i\in S$, $I(i,S_i^+)<1$. Let us fix a link $i$ in $S$. Then for each link $j\in S_i^+$, we have that $I(i,j)<1$. Namely, we have that $1 > I(i,j)=\frac{l_i^{\alpha}}{d(j,i)^{\alpha}}$. Since $\alpha>0$, this implies that $d(i,j) > l_i$, i.e.\  by definition, links $i,j$ are independent in $G_1$. Thus, $S$ is independent and $\chi(G_{1}(S))=1$. 
Hence, $\chi(G_1(T)) \le \sum_{k=1}^t \chi(G_1(S^k)) = t = O(1)$.
\end{proof}

\textit{Remark 1.} It is worth noting that the scheduling algorithms above do not require the underlying spanning tree to necessarily be an MST. Indeed the same argument would work also in the case when an arbitrary tree satisfying Lemma~\ref{L:mainlemma} is given. This can lead the way towards a definition of approximate MST that can have  efficient schedules.

\textit{Remark 2.} The result above can be extended to the case when one has stronger connectivity requirements. Namely, as it has been shown in~\cite{HMSODA12}, it is possible to construct a $k$-edge connected graph, for which Lemma~\ref{L:mainlemma} holds with $O(1)$ replaced with $O(k^4)$. Thus, our results hold also in the case when $k$-connectivity is required for any fixed $k$.

\subsection{Distributed Scheduling} \label{sec:distributed}

There is a natural adaptation of our coloring algorithms to distributed computation, which allows distributed computation of an aggregation schedule. Recall that the graphs $G_{arb}(L)$ and $G_{obl}(L)$ can be colored within constant approximation factors by considering the links in a non-increasing order by length and picking for each link the first color not yet used by its neighbor links in the corresponding conflict graph. We assume that the link set $L$ of an MST is already formed at the beginning of the algorithm, and that nodes have a polynomial upper bound on $\Delta$ and the number of links $n$. Each node knows the lengths of the links it participates in, as well as the minimum link length $l_{min}$ (or a common lower bound, up to constant factors) in the network.

The computation is done in $\lceil \log{\Delta}\rceil$ phases consisting of $polylog(n,\Delta)$ synchronized rounds. In phase $t=1,2,\dots,T$, only links in the length class $$L_t=\{i\in L: l_i\in [2^{(t-1)}l_{min}, 2^tl_{min}) \}$$ can transmit, the others remaining silent. The links use uniform power levels, proportional to the maximum link length in $L_t$. The computation starts from the class $L_T$, containing the longest link. First, a subroutine for computing constant factor approximate schedules for nearly equal length sets, i.e.\ length classes,  is run by nodes of links in $L_t$ (e.g.~\cite{Deltaplus1}). As soon as the links in $L_t$ establish a coloring, they locally broadcast their colors with uniform power level using a \emph{local broadcast} algorithm (e.g.~\cite{HMLocalBroadcast}). This way, the links in $L_t$ notify their shorter neighbors (including the ones in $L_t$) in in the graph about their color. Then, length class $L_{t+1}$ proceeds with the coloring algorithm, and so on.

There are many details to be taken into account for implementing and evaluating the algorithm, which depend on exact assumptions and model characteristics, so the analysis below should be taken with a grain of salt. The algorithms from~\cite{Deltaplus1} for coloring length classes of links run in time $O(opt_t\log{n})$, where $opt_t$ denotes the optimum schedule length for $L_t$. Note that in our case, $opt_t=\Theta(\log^* \Delta)$ or $opt_t=O(\log\log\Delta)$, depending on the power control mode. Crucially, the nodes (which, as assumed,  have a polynomial upper bound on $\Delta$) can compute an upper bound on these quantities, thus being able to pre-allocate a time for each phase and detect the end of each phase. Even though these algorithms compute a coloring from scratch, we believe it is possible to adapt them to take into account the set of colors used by earlier links, without degrading the runtime significantly. 
The local broadcast subroutine takes $O(opt_t + \log^2{n})$ rounds (with collision detection; 
$O(opt_t\cdot \log n + \log^2{n})$ rounds without it \cite{HMLocalBroadcast}), as the contention happens only between the links in $L_t$. Thus, realization of this scheme would give a distributed computation of schedules in time $O((\log{n}\cdot\log\log\Delta+\log^2{n})\log{\Delta})$ for the case of oblivious power schemes and in time $O((\log{n}\cdot\log^*\Delta+\log^2{n})\log{\Delta})$ for the case of global power control.

\section{Impossibility Results}\label{S:impossibility}

An optimal coloring schedule need not be an optimal aggregation schedule, as the optimum rate could be achieved by a schedule that consists of an arbitrary sequence of feasible sets. The latter is related to the notion of  \emph{multicoloring} or \emph{fractional coloring}.
A classic example where a multicoloring schedule improves over a coloring schedule is the (graph) coloring of the edges of the 5-cycle.
(This example can actually be mapped to an aggregation tree in the SINR model with $\beta = 1$.)
Any proper coloring requires three colors for a rate of $1/3$, but by using the sequence $13, 24, 14, 25, 35$ (where $13$ corresponds to the feasible set of the first and third edges in the cycle), we obtain a rate of $2/5$.

Thus, it is not obvious whether one can improve on the worst-case lower bounds on aggregation capacity we obtained  using coloring schedules.

We show that the bounds we obtained above are indeed  best possible. 
In particular, $1/\Theta(\log\log \Delta)$ is the best rate achievable using oblivious power assignments.
This holds even for instances located on the line.
For arbitrary power control, we also have a matching bound, $1/\Theta(\log^* \Delta)$, but only for the rate achievable on the MST of a pointset.

\subsection{Oblivious Power Schemes}\label{S:oblivious}

Fix a power scheme $P_{\tau}$ with $\tau\in (0,1)$. Let us denote $\tau' = \min\{\tau, 1-\tau\}$. Consider the doubly-exponential 
sequence of points  $1,2,\dots,n$ shown in Fig.~\ref{fig:exponentialchain}.
\begin{figure*}[ht]
\centering
\includegraphics[width=0.7\textwidth]{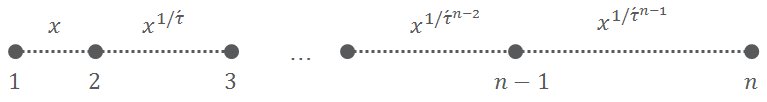}
\caption{A pointset of doubly-exponentially growing separation.}
\label{fig:exponentialchain}
\end{figure*}

The distance between points $t$ and $t+1$ is $x^{(\tau')^{-t}}$, for $t=1,2,\dots,n-1$, where $x> \min\{2, (2/\beta^{1/\alpha})^{1/\tau'}\}$ is a constant. Since $x\ge 2$, it is clear that the largest distance in the network is $\Theta(x^{(\tau')^{-n+1}})$, and the smallest distance is $x$. This implies that $n=\Theta(\log_{1/\tau'}(1+\log_{x}\log \Delta)=\Theta(\log \log \Delta)$.

We shall show that it is not possible to form two links on this point set that can be scheduled in the same slot. This implies that any aggregation schedule has to use a single link per time slot, implying an upper bound of 
$O(1/n)=O(1/\log\log \Delta)$ on the aggregation rate.

First, we define the relative interference of link $i$ on link $j$ under power assignment $P$, which is an additive operator 
$I_P(j,i)=\frac{\cI_{ji}}{\cS_i}=\frac{P(j)l_i^\alpha}{P(i)d_{ji}^\alpha}$. 
Let $I_P(i,i)=0$ and $I_P(S,i)=\sum_{j\in S}{I_P(j,i)}$ for simplicity.  
In absence of the noise term, the feasibility condition of a set $S$ of links, for a fixed
power assignment $P$, can be rewritten as $I_P(S,i) \le 1/\beta$.  When using a power scheme $P_\tau$, we have
$I_{P_\tau}(j,i)= l_j^{\tau\alpha} l_i^{(1-\tau)\alpha}/d_{ji}^\alpha$.

Consider any two links $1$ and $2$. Assume, w.l.o.g., that $2$ has the greater right endpoint among the two, and that this endpoint is the point $t+1$. That  implies that $l_2\ge x^{(\tau')^{-t}}$. If the link $1$ also uses point $t+1$, then the two links clearly can't be scheduled in the same slot. Otherwise, we have that $l_1\le x^{(\tau')^{-t+1}} \le l_2^\tau$. We consider two cases:\\
1. If the longer link $2$ is directed left to right, then it interferes with the shorter link. Since the sender node of link $2$ is to the left of point $t+1$, and the set of nodes forms a exponentially increasing sequence, its distance to any of the nodes in $\{1,2,\dots,t\}$ is at most twice the distance between points $t-1$ and $t$. Thus, we have $d_{21}\le 2 x^{(\tau')^{-t+1}}\le 2l_2^{\tau}$ (since $\tau'\le \tau$). We have further $l_1\ge x > (2/\beta^{1/\alpha})^{1/(1-\tau)}$ and
\[I_{P_{\tau}}(2, 1)=\frac{l_2^{\tau\alpha}l_1^{(1-\tau)\alpha}}{d_{21}^\alpha} \ge \frac{l_2^{\tau\alpha}l_1^{(1-\tau)\alpha}}{(2 l_2^\tau)^\alpha}=\frac{l_1^{(1-\tau)\alpha}}{2^\alpha} > 1/\beta.\]
2. If the longer link $2$ is directed right to left, then the shorter link $1$ interferes with the longer one. This time we have $d_{12}\le 2 x^{(\tau')^{-t+1}}\le 2l_2^{1-\tau}$ (since $\tau'\le 1-\tau$) and $l_1\ge x > (2/\beta^{1/\alpha})^{1/\tau}$, so
\[I_{P_{\tau}}(1, 2)=\frac{l_1^{\tau\alpha}l_2^{(1-\tau)\alpha}}{d_{21}^\alpha} \ge \frac{l_1^{\tau\alpha}l_2^{(1-\tau)\alpha}}{(2 l_2^{1-\tau})^\alpha}=\frac{l_1^{\tau\alpha}}{2^\alpha} > 1/\beta\]
Thus, links $1$ and $2$ are incompatible, so we obtain the lower bound as desired. This is summarized below.

\begin{proposition}
For any $\tau\in (0,1)$ and $\Delta>0$, there is a set $S$ of points on the real line with $\Delta(S)>\Delta$, for which every aggregation tree and schedule gives rate $\Theta(1/\log{\log{\Delta(S)}})$  when using $P_{\tau}$.
\end{proposition}

\subsection{Arbitrary Power Control}

We now show that there are instances on the real line whose MST cannot be aggregated with rate better than $O(1/\log^* \Delta)$.
The construction is a modification of a construction in \cite[Thm.~7]{HMSODA12}, that leads only to a lower bound of $O(1/\log(\log^*\Delta))$ and applies only to colorings, not general aggregation schedules.

We will use the following theorem from~\cite{us:stoc15} in the proof. Recall the additive operator $I$ defined in Sec.~\ref{sec:bounding}: $I(j,i)=\min \{1, l_j^\alpha/d(i,j)^\alpha\}$, for links $i,j$. 

\begin{theorem}\cite[Thm. 8]{us:stoc15}\label{T:necessary}
Let $S$ be a set of links that is feasible with $\beta=3^\alpha$. Then $I(S_i^-,i) = O(1)$ for each $i\in S$.
\end{theorem}

\begin{theorem}\label{T:hardmst}
For any $\Delta$, there is a pointset on the real line whose minimum spanning tree $T$  has optimal aggregation rate of $O(1/\log^* \Delta(T))$
and $\Delta(T)>\Delta$.
\end{theorem}
\begin{proof}
 Recall that for any set of nodes at different points on the line, the edges of their unique MST are obtained by connecting each node to the closest nodes from both sides. In order to not complicate the notation, for each set $R$ of nodes, we  identify $R$ with the MST over $R$. We recursively construct instances $R_t$ for $t=1,2,\dots$ such that $t=\Omega(\log^*{\Delta(R_t)})$ and $R_t$ cannot be scheduled using less than $t$ slots. For simplicity, we assume that $\beta=3^\alpha$. Note, however, that the bound can be extended to arbitrary $\beta\ge 1$: if there was an aggregation schedule with rate $r$ for $\beta=1$, we could transform it into a schedule of rate $\Omega(r)$ by partitioning every feasible set in the schedule into a constant number of subsets, each feasible with $\beta=3^\alpha$ (see, e.g. \cite{GoussevskaiaHW14}).

Given the instance $R_t$, we construct $R_{t+1}$  by joining together different scaled copies of $R_t$ and adding another long link $\ell$ to this construction. The sub-instances $R_t$ are scaled and placed so that the following property holds: if one chooses any collection of links $i_1, i_2,\dots,i_s$, one link from each copy of $R_t$, then at most half that set can be scheduled in the same slot as the long link $\ell$. We will then use this property to inductively prove the rate bound.

We will use the following notations. For an instance $R$, we will use $diam(R)$ to denote the diameter of $R$, i.e.\ the maximum distance between nodes in $R$.
For two instances $R$ and $G$, $R\oplus G$ denotes the joining of the two instances, which is a new instance on $|G|+|R|- 1$ nodes, having $R$ on the left side and $G$ on the right side, where the two sub-instances have one common node.
Note that the edges of the MST of $R\oplus G$ is the union of the edges of MSTs of $R$ and $G$.
For an instance $R$ and a link $i\in R$, $\hat{d}_i(R)$ is the maximum distance from either endpoint of $i$ to the leftmost point of $R$. We also denote \[\rho(R)=\min_{i\in R}{\frac{l_i^\alpha}{\hat{d}_i(R)^\alpha}}.\] Note that for each $R$, $\rho(R)\leq 1$.

Now we are ready to describe the construction. The instance $R_1$ consists of two nodes at a distance $1$.  Let us assume that the instance $R_t$ for $t\ge 1$ is constructed with the desired properties. Then $R_{t+1}$ is constructed using many scaled copies of $R_t$ (see Fig.~\ref{fig:mstconstruction}).
 First we define $R'_{t+1}=\oplus_{1\leq s\leq k_{t+1}}{R_{t}^s}$, which is the concatenation of $k_{t+1}=\frac{c}{\rho(R_{t})}$ scaled copies of $R_t$ such that  $R_{t}^1$ is an identical copy of $R_{t}$ and for each $s>1$, $R_{t}^s$ is a copy of $R_{t}$, scaled so that the maximum link length of $R_{t}^s$ is equal to $diam(\oplus_{1\leq r\leq s-1}{R_{t}^r})$, and $c$ is a large enough constant, determined in Claim~\ref{C:mstX}. Then we define $R_{t+1}=G\oplus R'_{t+1}$ where $G$ is an instance consisting of two points at a distance $diam(R'_{t+1})$. This completes the construction of $R_{t+1}$. 
\begin{figure*}[htbp]
\centering
\includegraphics[width=0.7\textwidth]{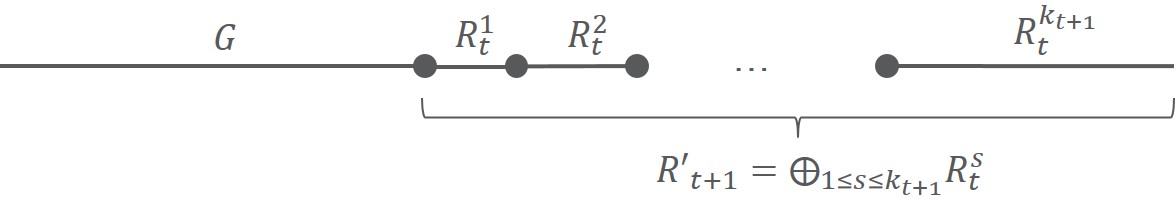}
\caption{The construction of $R_{t+1}$.}
\label{fig:mstconstruction}
\end{figure*}

We now show that the aggregation rate that can be attained on the MST $T$ of $R_{t+1}$ is at most $2/(t+1)$.
%
The proof is by induction on $t$. The statement is trivially true for the base case $t=1$.
Suppose the statement holds for a given $t$ and consider the case of $t+1$.
Consider an arbitrary aggregation schedule $S_1,S_2,S_3,\dots$ of $R_{t+1}$. Assume that there is an integer $m_0$, such that for every $m\ge m_0$, the long link $j$ is scheduled in at least $2m/(t+1)$ slots among the first $m$ slots. If there is no such $m_0$ then we are done, as the schedule does not achieve rate $2/(t+1)$. Consider an arbitrary integer $m\ge m_0$.

Let $T_1, T_2,\ldots, T_k$ be the feasible sets from $S_1,S_2,\dots,S_m$ containing the link $j$. Since the rate $2/(t+1)$ is achieved on $j$, 
we have $k\ge 2m/(t+1)$.
Recall that $R_{t+1}$ consists of $j$ and $k_{t+1}$ copies of $R_{t}$.

\begin{claim}\label{C:mstX}
$|T_q| \le k_{t+1}/2$ for each $q=1,2,\ldots, k$.
\end{claim}

\begin{proof}
By construction, $j$ is the longest link in $R_{t+1}$.
Let us fix a link $i\in R_{t}^s$ for some $s\geq 1$. Observe that 
$
d(i,j)\le \hat{d}_i(R_{t}^s) + diam(\oplus_{1\leq r\leq s-1}{R_{t}^r}). 
$
By construction, $\hat{d}_i(R_{t}^s) \ge l_{j_s} =diam(\oplus_{1\leq r\leq s-1}{R_{t}^r})$, where $j_s$ is the longest link in $R_{t}^s$. This follows from the fact that $j_s$ is the leftmost link in $R_{t}^s$. Thus, we have that $d(i,j) \le 2 \hat{d}_i(R_{t}^s)$ and so 
\[ I(i,j)=\frac{l_i^{\alpha}}{d(i,j)^\alpha} \ge \frac{l_i^\alpha}{2\hat{d}_i(R_{t}^s)}
\ge \frac{\rho(R_{t}^s)}{2^{\alpha}} =\frac{\rho(R_t)}{2^\alpha}\ , \]
where the last equality follows from the fact that $\rho(R_t)$ is scale-invariant. 
Now, if there was a feasible set $T$ containing link $j$ and $k_{t+1}/2$ other links, we would have 
$I(T_i^-, i) \ge k_{t+1}/2 \cdot \rho(R_t)/2^\alpha = c/2^{\alpha+1}$, 
which would contradict Thm.~\ref{T:necessary}, if we choose $c$ sufficiently large.
\end{proof}

Therefore, $\sum_{q=1}^{k} |T_q| \le k\cdot k_{t+1}/2$, 
which implies there is a copy $R^{s}_t$ containing links from at most $k/2$.
In other words, at least $m-k/2$ of the sets $T_q$ contain no link of copy $R_t^s$. 
This means that at most $k/2$ sets $T_q$ contain a link from $R_t^s$, and conversely there are at least $k/2$ sets $T_q$ that contain no link from $R_t^s$.
Thus, there are at most $m-k/2$ sets in the original sequence $S_1,\dots,S_m$ that can contain a link in $R_t^s$. 
By induction, there is no schedule with rate more than $2/t$ for $R_t^s$, so there is a link $i\in R_t^s$ that appears in at most $2/t$ fraction of those sets, i.e.\ in at most $\frac{2}{t}(m-k/2)=2m/(t+1)$ sets in the sequence  $S_1,\dots,S_m$. Since this holds for arbitrary $m\ge m_0$, this proves that $2/(t+1)$ is the best possible rate.

\smallskip

It remains to show that $t=\Omega(\log^*{\Delta(R_t)})$. Note that the longest link of $R_{t}$ is the link including the leftmost point of $R_{t}$ and the shortest link has length 1.
That implies that $1/\rho(R_{t}) \leq \Delta(R_{t})$; hence, $k_{t}= c/\rho(R_{t-1}) \le c\Delta(R_{t-1})$. Note also that for each $s>1$, $diam(R_{t-1}^s) \le 2diam(R_{t-1}^{s-1})$. It follows that 
\begin{align*}
\Delta(R_t)&=2diam(R'_t)
=2\sum_{1}^{k_t}{diam(R_{t-1}^s)}\\
&\leq 2diam(R_{t-1})\sum_{1}^{k_t}{2^s}
\leq 4\Delta(R_{t-1})\cdot 2^{k_t}\\
&\le a^{\Delta(R_{t-1})},
\end{align*}
where $a$ is a constant. Thus, $\Delta(R_{t-1})= \Omega(\log\Delta(R_t))$, which implies that $t= \Omega(\log^*{\Delta(R_t)})$, recalling that $\Delta(R_1)=1$. This completes the proof.
\end{proof}

\textit{Remark.} Note the difference between the bounds obtained for the case of oblivious power schemes and global power control. Namely, in terms of the number of nodes $n$, the bound for the case of oblivious power schemes is $O(1/n)$, while in the case of global power control, it is only $O(1/\log^* n)$. While the former bound clearly cannot be improved, the latter is far from the known lower bound $\Omega(1/\log n)$ in terms of $n$ alone.

\section{Is MST Optimal for Aggregation?}
\label{sec:mstoptimal}

We showed that an MST of any set of nodes has a short coloring schedule. However, we saw that the length of the schedule is not constant in general, although grows very slowly with the size of the network. A natural question from the theoretical perspective is, whether one can find other spanning trees that can be colored using even less slots. In particular, intuition tells us that at least for linear networks, MST should be optimal (possibly up to constant factors). 

This intuition is supported by the case of uniform and linear power schemes ($P_0$ and $P_1$).  Let us assume that a threshold value $\beta\ge 3^\alpha$ is used (this assumption affects only constant factors \cite{GoussevskaiaHW14}). This choice of $\beta$ ensures that any two links in a $P_0$-(or $P_1$-)feasible set of links must be separated by a distance more than the length of the longer link. For example, if $i,j$ are in the same $P_0$-feasible set, then $d_{ij} > \beta^{1/\alpha} l_j\ge 3l_j$ and similarly, $d_{ji} > 3l_i$. This implies, using the triangle inequality, that the minimum distance between links $i,j$ is at least $\max\{l_i,l_j\}$. Consider a set of nodes arranged on the line and an arbitrary spanning tree $T$ of this set. This means that links in the same feasible set must be non-overlapping. On the other hand, since MST always connects closest pairs of nodes on the line, one can make a correspondence between the edges of $T$ and the MST, such that by the non-overlapping property above, for each feasible subset of $T$, the  corresponding subset of the MST is also feasible.
It is not difficult to make this intuitive argument rigorous and show that for the uniform and linear power schemes the MST is optimal  on the line (modulo constant factors).

\begin{proposition}
For any set of nodes on the line, the MST is a constant-factor approximate solution for the aggregation scheduling problem for power schemes $P_{0}$ and $P_1$.
\end{proposition}

Does this hold for all other oblivious power schemes? A straightforward check reveals that for the examples considered in Sec. \ref{S:oblivious}, the MST is optimal, as one cannot construct \emph{any} $P_{\tau}$-feasible pair of links on the given set of nodes. However, despite these supporting facts, it turns out that the intuition may be deceptive in this case. 

Below, we present a family of network instances on the real line together with a spanning tree $T$, such that $T$ has a constant length coloring schedule, but the MST cannot be colored with less than $\Theta(\log\log{\Delta})=\Theta(n)$ slots. This holds for any $\tau \in (0,2/5]\cup [3/5,1)$.

First, assume that $\tau \le 2/5$.
We explain the construction using eight nodes on the real line, but the argument is general and is straightforwardly extended to obtain an infinite family of networks with similar properties. Consider the network consisting of 8 nodes and the spanning tree shown in Fig.~\ref{fig:mstsubopt}.
\begin{figure}[htbp]
\centering
\includegraphics[width=0.47\textwidth]{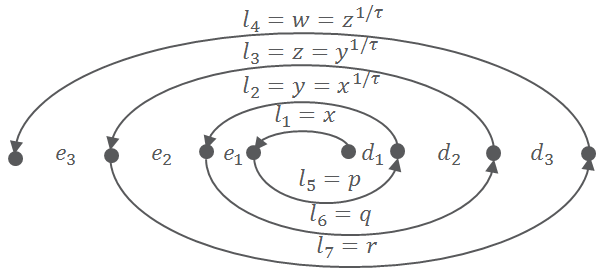}
\caption{A network showing sub-optimality of MST on the line.}
\label{fig:mstsubopt}
\end{figure}
Note that the given set of links can be augmented by a single link to yield a strongly connected graph. The shortest link has length $l_1=x$, where $x$ is a large enough constant. The lengths of other links are defined in terms of $x$, as follows: 
\begin{align*}
l_2=y &= x^{1/\tau},  &&l_5=p = y^{\tau}x^{1-\tau + \tau^2}\\
l_3=z &= y^{1/\tau},  &&l_6=q = z^{\tau}y^{1-\tau + \tau^2}\\
l_4=w & = z^{1/\tau}, &&l_7=r  = w^{\tau}z^{1-\tau + \tau^2}.
\end{align*}
All other distances in the network are determined through the lengths of the links.

\begin{claim} 
The sets $S=\{1,2,3,4\}$ and $S' = \{5,6,7\}$ are $P_\tau$-feasible.
\end{claim}
\begin{proof}
We first show that $S = \{1,2,3,4\}$ is $P_\tau$-feasible.
Note that if $x$ is large enough and $\tau \le 2/5$ then 
\[p =y^\tau \cdot (y^\tau)^{1-\tau + \tau^2} \le y^{7/9}< y/2,\] 
which implies that $e_1=d(r_1,r_2) = y - p > y/2.$ Similarly, we have that $e_2 > z/2$ and $e_3 > w/2$. Thus, by denoting $h=x^{\tau\alpha}$, we have
\begin{align*}
I_{P_{\tau}}(1, 4) &=\frac{x^{\tau\alpha} w^{(1-\tau)\alpha}}{d_{14}^\alpha}
< \frac{x^{\tau\alpha} w^{(1-\tau)\alpha}}{e_3^\alpha}\\
&\le 2^\alpha(x/w)^{\tau\alpha}=2^\alpha h^{1-1/\tau^3}.
\end{align*}
Similarly, using the estimates for $e_2$ and $e_1$, we obtain the bounds $I_{P_{\tau}}(2, 4) < 2^\alpha(y/w)^{\tau\alpha}=2^\alpha h^{1-1/\tau^2}$ and $I_{P_{\tau}}(3, 4) < 2^\alpha h^{1-1/\tau}$, respectively. Thus,
\[
I_{P_\tau}(S,4)< 2^\alpha h \cdot (h^{-1/\tau} + h^{-1/\tau^2} + h^{-1/\tau^3}).
\]
If $\tau\le 1/2$ and $x$ is large enough (depending on $\beta$) then the right hand side is less than $1/\beta$, as the terms of the sum in the parentheses decrease doubly-exponentially, and largest of them is at most $h^{-2}$.

Now let us bound $I_{P_\tau}(S,1)$. Note that $d_{21}=p$, $d_{31} = q - e_1$ and $d_{41}=r - e_1 - e_2$. Recall that $e_k =\Theta(l_{k+1})$ for $k=1,2,3$, which implies that the sequence $e_1, e_2, e_3$ grows doubly exponentially. Thus, if $x$ is large enough, we have that $e_2 > e_1$ and $e_3 > e_2 + e_1$ and $d_{31} = q - e_1 > q - y > q/2$ and $d_{41} > r - e_3 > r - z > r/2$. Thus, we have 
\[I_{P_{\tau}}(2, 1) = \frac{x^{(1-\tau)\alpha}y^{\tau\alpha}}{d_{21}^\alpha} =  \frac{x^{(1-\tau)\alpha}y^{\tau\alpha}}{p^\alpha}=x^{-\tau^2\alpha}\]
 and similarly, using the bounds above, 
 \[
 I_{P_{\tau}}(3, 1)=\frac{y^{(1-\tau)\alpha}z^{\tau\alpha}}{d_{31}^\alpha} < 2^\alpha y^{-\tau^2\alpha}
 \] and $I_{P_{\tau}}(4, 1) < 2^\alpha z^{-\tau^2\alpha}$. Combining these bounds, we see that $I_{P_\tau}(S,4)$ can be easily bounded by $1/\beta$ when $x$ is a large enough constant.
The values $I_{P_\tau}(S,2)$ and $I_{P_\tau}(S,3)$ can be bounded similarly, implying that the set $S=\{1,2,3,4\}$ is feasible.

Now let us consider the set $S'=\{5,6,7\}$. First, note that $d_2 = q - y > q/2$ and $d_3 = r - z > r/2$ when $x$ is large enough. In particular, $d_{57} > d_3 > r/2$ and $d_{67} > d_3 > r/2$. Thus, we have that 
\[I_{P_{\tau}}(5,7) < \frac{p^{\tau\alpha}r^{(1-\tau)\alpha}}{d_3^\alpha} < 2^\alpha (p/r)^{\tau\alpha}\]
 and 
 \[I_{P_{\tau}}(6,7)<\frac{q^{\tau\alpha}r^{(1-\tau)\alpha}}{d_3^\alpha} < 2^\alpha (q/r)^{\tau\alpha}.\]
  Thus, $I_{P_{\tau}}(S',7)$  can be bounded by $1/\beta$. Now consider link $5$. Note that $d_{65}= y$ and $d_{75}= z$. Then we have 
\begin{align*}
I_{P_{\tau}}(6,5)&< \frac{q^{\tau\alpha}p^{(1-\tau)\alpha}}{y^\alpha}\\
&=\frac{y^{(2-\tau + \tau^2)\tau\alpha}\cdot y^{(2-\tau + \tau^2)(1-\tau)\tau\alpha}}{y^\alpha}< y^{-\gamma\alpha},
\end{align*}
where $\gamma = 1-4\tau + 4\tau^2 - 3\tau^3 + \tau^4> 0$ when $\tau \le 2/5$.
We have similarly $a_{P_{\tau}}(7, 5) < z^{-\gamma\alpha}$. Thus, again, $I_{P_\tau}(S',5)$ is bounded  by a sum of doubly-exponentially decreasing terms, which can be bounded by $1/\beta$.
\end{proof}
 
Thus, we conclude that  the set of links $1,2,\dots,7$ can be scheduled in two slots. 
On the other hand, if we consider the (unique) minimum spanning tree of the set of given nodes, then the subtree corresponding to the intervals $e_1, e_2, e_3$ will take exactly as many slots as it has links, by the result of Sec. \ref{S:oblivious}, since this subtree corresponds to the instance of
doubly-exponential network constructed in Sec. \ref{S:oblivious} (with powers of $1/\tau$ in the exponent) and $\min\{\tau, 1-\tau\} = \tau$.

A symmetric argument will apply when $\tau \ge 1-2/5=3/5$. In this case, we reverse the directions of the links and use the following new definition of link lengths:
\begin{align*}
l_1=y &= x^{1/(1-\tau)},  &&l_4=p = y^{1-\tau}x^{\tau + (1-\tau)^2}\\
l_2=z &= y^{1/(1-\tau)},  &&l_5=q = z^{1-\tau}y^{\tau + (1-\tau)^2}\\
l_3=w & = z^{1/(1-\tau)}, &&l_6=r  = w^{1-\tau}z^{\tau + (1-\tau)^2}.
\end{align*}
 The construction and proof can clearly be extended for an arbitrary number of nodes on the line. Thus, we essentially proved the following:
\begin{proposition}
For each $\tau\in (0,2/5]\cup [3/5, 1)$, there is a family of linear networks such that the optimum solution of aggregation scheduling using $P_{\tau}$ takes a constant number of slots, but the MST cannot be scheduled in less than $\Theta(\log\log\Delta)=\Theta(n)$ slots using $P_{\tau}$.
\end{proposition}

We were not able to prove similar results for the case of global power control. The main difficulty stems from the fact that the  lower bound constructions we have for scheduling MST are very complicated, which makes it  difficult to ``rewire'' them, in order to obtain more efficiently schedulable spanning trees, if such exist.

\section{Conclusion}

We addressed the problem of estimating best aggregation capacity  in sensor networks. In our setting, it essentially boils down to finding minimally schedulable spanning tree of a set of wireless nodes. For this purpose, we consider an MST and show that with appropriate power control, the logarithmic bounds achieved in the previous work can be replaced with $\log^*$- or double-logarithmic bounds. This demonstrates once again the importance of right power control. We also obtained matching lower bounds.

We concentrated on the rate maximization aspect of aggregation. It would be interesting to see what tradeoffs can be achieved between rate maximization and delay minimization. Our results hold for perfectly compressible functions, but as mentioned in Sec.~\ref{S:aggregationprotocol}, it is possible to use our techniques for other functions. We leave more concrete results on this direction to the future work.

\bibliographystyle{abbrv}
\bibliography{Bibliography}

\appendix

\section{Conflict Graphs}
\label{S:conflict}

Given a set $L$ of links, one can define a binary notion of conflict between pairs of links. This naturally leads to a graph representation, i.e. a conflict graph. We will be interested in the following family of conflict graphs, introduced in~\cite{us:stoc15}. We consider only simple graphs.

Let $f:[1,\infty)\rightarrow \mathbb{R}_+$ be a positive non-decreasing sub-linear function.
Two links $i,j$ are said to be \emph{$f$-independent} if
  \[ \frac{d(i,j)}{l_{min}} > f\left(\frac{l_{max}}{l_{min}}\right), \]
where $l_{min}=\min\{l_i,l_j\},l_{max}=\max\{l_i,l_j\}$, and otherwise they are \emph{$f$-conflicting}.
A set of links is $f$-independent if they are pairwise $f$-independent.
The conflict graph $\cG_f(L)$\label{G:gf}, for a set $L$ of links, is the graph with vertex set $L$, where two vertices $i,j\in L$ are adjacent if and only if they are $f$-conflicting.

The usefulness of these conflict graphs stems from the fact that they give a good approximation for feasibility. Moreover they can be used to compute efficient schedules, as shown in~\cite{us:stoc15, us:fsttcs15}.

We will use only the following special cases. Below, $\gamma >0$ is a positive real parameter.
\begin{itemize}
\item{$f(x)\equiv \gamma$. Let $\cG_{\gamma}(L)$ denote this conflict graph.}
\item{$f(x)= \gamma x^\delta$ with $\delta \in (0,1)$. Let $\cG_{\gamma}^\delta(L)$ denote the conflict graph.}
\item{$f(x) = \gamma \cdot\max\{1, \log^{2/(\alpha - 2)}x\}$. Let $\cG_{\gamma\log}(L)$ denote the conflict graph.}
\end{itemize}

Below, we list some properties of the graphs above that we will use. The proofs can be found in~\cite{us:stoc15, us:fsttcs15}. We let $\chi(G)$\label{G:chi} denote the \emph{chromatic number} of a graph $G$, the minimum number of colors needed for coloring the vertices of $G$ such that adjacent vertices get different colors.
The following holds for any set $L$ of links.
\begin{enumerate}[a.]
\item{If constant $\gamma_1$ is large enough then each independent set in $\cG_{\gamma_1\tlog}(L)$ is feasible \cite[Cor. 1]{us:stoc15}.
There is an oblivious power scheme $P_{\tau}$, such that for an appropriate constant $\delta\in (0,1)$ and large enough constant $\gamma_2$, any independent set in $\cG_{\gamma_2}^\delta (L)$ is $P_{\tau}$-feasible \cite[Cor. 6]{us:fsttcs15}.}\label{I:upperb}
\item{For any constants $\gamma, \gamma' > 0$ and $\delta\in (0,1)$, 
\begin{align*}
&\chi(\cG_{\gamma}^\delta(L))= \chi(\cG_{\gamma'}(L))\cdot O(\log\log\Delta)\mbox{ and}\\
&\chi(\cG_{\gamma\tlog}(L))= \chi(\cG_{\gamma'}(L))\cdot O(\log^*\Delta),
\end{align*}
where $\Delta=\Delta(L)$ \cite[Thm. 1]{us:stoc15}.
}\label{I:gap}
\item{There is a constant factor approximation algorithm for vertex coloring in graphs $\cG_f$ with sub-linear function $f$. This holds, in particular, for graphs $\cG_{\gamma}$, $\cG_{\gamma}^\delta$ and $\cG_{\gamma\tlog}$.}\label{I:algo}
\end{enumerate}

The first property shows that coloring the graphs $\cG_{\gamma_1\log}$ and $\cG_{\gamma_2}^\delta$ gives coloring schedules for the set $L$. The second property property gives upper bounds for the number of slots needed.
The third property shows that the described upper bounds can be efficiently computed. This is achieved by proving the following property of $\cG_f$. Consider any link $i\in L$. Let $N^+_i$ denote the set of links that are not shorter than link $i$ and conflict with $i$. As shown in~\cite{us:stoc15}, the  cardinality of \emph{any} independent subset of $N_i^+$ is bounded by a constant. In other words, $\cG_f$ has constant \emph{inductive independence}, as defined in~\cite{yeborodin}. It is easy to check then that the following greedy algorithm is a constant factor approximation algorithm for coloring $\cG_f(L)$ (see e.g.~\cite{yeborodin} for details): consider the links in a non-increasing order by length, and assign each link the first available color, i.e. the one that has not been used by its neighbors yet.

Thus, to conclude this section, we restore the notation for Sec.~\ref{S:aggregationprotocol}: we set $G_1(L)=\cG_\gamma(L)$ for a constant $\gamma$ (whose choice affects only the constant factors in the approximation ratio), $G_{arb}(L)=\cG_{\gamma_1\log}(L)$ and $G_{obl}=\cG_{\gamma_2}^\delta$, for appropriate parameters of $\gamma_1, \gamma_2, \delta$.

\end{document}